\documentclass[a4paper, 10pt, notitlepage]{article}

\usepackage{amsmath, amsthm, amsfonts, tikz, hyperref, graphics, authblk, mathdots,soul}
\usepackage[font=small, margin=10mm, labelfont=bf]{caption}
\DeclareMathOperator{\Tr}{Tr}
\DeclareMathOperator{\e}{e}
\DeclareMathOperator{\im}{i}
\DeclareMathOperator{\di}{d}

\newcommand{\lab}{(inv)}

\newtheorem{theorem}{Theorem}
\newtheorem{corollary}{Corollary}
\newtheorem{lemma}{Lemma}

\numberwithin{equation}{section}

\title{Spectra and eigenstates of spin chain Hamiltonians}

\author[1]{J.P. Keating\thanks{j.p.keating@bristol.ac.uk}}
\author[1]{N. Linden\thanks{n.linden@bristol.ac.uk}}
\author[1]{H.J. Wells\thanks{huw.wells@bristol.ac.uk}}
\affil[1]{School of Mathematics, University of Bristol, Bristol, BS8 1TW, UK}

\begin{document}
\date{March 2014}
\maketitle

\begin{abstract}
We prove that translationally invariant Hamiltonians of a chain of $n$ qubits with nearest-neighbour interactions have two seemingly contradictory features.  Firstly in the limit $n\rightarrow\infty$ we show that any translationally invariant Hamiltonian of a chain of $n$ qubits has an eigenbasis such that almost all eigenstates have maximal entanglement between fixed-size sub-blocks of qubits and the rest of the system; in this sense these eigenstates are like those of completely general Hamiltonians (i.e. Hamiltonians with interactions of all orders between arbitrary groups of qubits).  Secondly in the limit $n\rightarrow\infty$ we show that any nearest-neighbour Hamiltonian of a chain of $n$ qubits 
has a Gaussian density of states; thus as far as the eigenvalues are concerned the system is like a non-interacting one.   The comparison applies to chains of qubits with translationally invariant nearest-neighbour interactions, but we show that it is extendible to much more general systems (both in terms of the local dimension and the geometry of interaction).  Numerical evidence is also presented which suggests that the translational invariance condition may be dropped in the case of nearest-neighbour chains.
\end{abstract}

\section{Introduction}
Quantum spin chains are ubiquitous in modern physics.  Since Dirac and Heisenberg first observed a quantum ferromagnetic phase transition in the ground state of such a model in 1926 \cite{Mattis1988}, solving the long standing problem of ferromagnetism, their use has been widespread.  These models display a rich variety of quantum phase transitions \cite{Quantum} and have applications to high fidelity quantum state transfer \cite{Bose2003,Linden2004}.

Entanglement in the ground state of quantum spin chains, between a continuous block of $l$ spins and the rest of the chain, has also been widely studied by many authors \cite{Vidal2002,Vidal2004,Korepin2003,Korepin2004,Korepin2005}, via the entropy of entanglement; and techniques from random matrix theory have proved very productive \cite{Keating2004,Keating2005}.  Here, the entropy is logarithmic in $l$ for a wide range of critical chains.  Masanes \cite{Masanes2010} proved an area law for the entropy of entanglement of low energy eigenstates of a spatially-extended quantum system with local interactions.  Entanglement in the low energy eigenstates of chains has also been considered in \cite{Vedral2000,Vedral2001}, and for higher eigenstates for integrable models in \cite{Cardy2004}.

The entanglement of low lying eigenstates of nearest-neighbour Hamiltonians is in marked contrast to that of the eigenstates of completely general Hamiltonians (we can think of such a system as having interactions between all groups of parties of all orders).   Such an eigenstate, being a random state in the full Hilbert space, has entropy of entanglement proportional to the number of qubits $l$ in the limit $n\rightarrow\infty$ \cite{Winter2006}.

At the other end of the spectrum are completely local Hamiltonians (given by sums of terms involving only one qubit).  Such Hamiltonians have eigenstates that are product states, and a spectral density that is Gaussian in the large $n$ limit.

In this paper we identify two seemingly contradictory properties of translationally invariant Hamiltonians of $n$ qubits with nearest-neighbour interactions.  Firstly we prove that for general transitionally invariant Hamiltonians of a chain of $n$ qubits, almost all eigenstates have maximal entanglement between sub-blocks of $l$ qubits and the rest of the system in the limit $n\rightarrow\infty$.  In particular this applies to nearest-neighbour translationally invariant Hamiltonians.  Secondly we show that the density of states of such nearest-neighbour systems is Gaussian. Thus, even though the nearest-neighbour Hamiltonian is local and thus specified by a small number of parameters, loosely speaking, almost all eigenstates are rather like (at least if one looks at relatively small blocks of spins) those of completely general Hamiltonians but the distribution of eigenvalues is similar to a non-interacting system.  In fact we will prove that  the density of states has a Gaussian distribution even in the non-translationally invariant case.

While the results are initially proven for chains of qubits, we show that many of the results are extendible to much more general systems (both in terms of the local dimension and the geometry of interaction).

In quantum statistical mechanics, it is a fundamental question to understand the mechanism of equilibration \cite{Linden2009,Goldstein2010}. One idea is that it can be traced directly to properties of the eigenstates: the celebrated \lq\lq eigenstate thermalization\rq\rq\ hypothesis \cite{Deutsch1991,Srednicki1994}.  Although very attractive, and borne out in some examples \cite{Rigol2008}, this hypothesis has so far resisted proof.  While we do not prove the hypothesis, our work can be understood as providing evidence towards it: our results show that the infinite temperature eigenstates (which constitute almost all states) obey the hypothesis for a wide class of physically interesting systems.

This paper is split into the following sections:  In Section \ref{model} the two basic classes of nearest-neighbour spin chain Hamiltonians which exhibit the seemingly contradictory features of interest are introduced.  The first is a generic form of a spin chain Hamiltonian describing a ring of $n$ qubits each interacting with only their nearest-neighbours.  The second is a subclass of this, containing only those Hamiltonians which are invariant with respect to translation around the ring.

In Section \ref{results} the main results of the paper are given.  Theorem \ref{th2} and its corollary show that the purity of most reduced eigenstates of a general translationally invariant qubit chain Hamiltonian (including those with nearest-neighbour interactions), on a continuous block of an asymptotically small proportion of the total number of qubits, tends towards its minimal value.   That is, most eigenstates are maximally entangled between these qubits and the rest of the chain.    The proof of this theorem is given in Section \ref{Th2proof}; it only requires the translational invariance of the system.  Extensions to some non-translationally invariant chains, higher dimensional systems and the effect of eigenstate degeneracy are also given.  In section \ref{dis} numerical evidence is presented which suggests that the translational invariance condition may be dropped in the case of nearest-neighbour interactions.

When applied to nearest-neighbour translationally invariant qubit chains, Theorem \ref{th2} is in contrast to Theorem \ref{DOS} which shows that a central limit theorem holds for the spectrum of many general nearest-neighbour qubit chain Hamiltonians (including many translationally invariant ones), in the large-chain limit.  This limiting spectral behaviour is exactly that seen for generic Hamiltonians of systems of non-interacting qubits.  The proof of this theorem is given in Section \ref{Th1proof} and its generalisations to qudits, more general interaction geometries and a proof of a conjecture by Atas and Bogomolny \cite{Atas2014} are established.

The crossover between these interacting and non-interacting type behaviours is also particularly sensitive.  Given two large portions of a chain, the size of the terms connecting them in the Hamiltonian, the `interaction' term, can be arbitrarily small compared to the rest of the Hamiltonian and yet still affect the eigenstates of the system dramatically.

Section \ref{dis} outlines some open problems.

\section{Spin chain Hamiltonians}\label{model}
The following two classes of spin Hamiltonians, describing a chain or ring of $n$ qubits with nearest-neighbour interactions, will be used to demonstrate seemingly contradictory features present within their spectral decompositions.  The extensions to qudits and more general interaction geometries will be made where appropriate.

The first class contains Hamiltonians of the form
\begin{equation}
  H_n=\frac{1}{\sqrt{n}}\sum_{j=1}^n\sum_{a=0}^3\sum_{b=1}^3\alpha_{a,b,j}\sigma_j^{(a)}\sigma_{j+1}^{(b)}
\end{equation}
for any $\alpha_{a,b,j}=\alpha_{a,b,j}(n)\in\mathbb{R}$ and the matrices
\begin{equation}
  \sigma_{  j  }^{(a)}=I_2^{\otimes(j-1)}\otimes\sigma^{(a)}\otimes I_2^{\otimes(n-j)}
\end{equation}
where $\sigma^{(1)}$, $\sigma^{(2)}$ and $\sigma^{(3)}$ are the $2\times 2$ Pauli matrices and $I_2\equiv\sigma^{(0)}$ is the $2\times2$ identity matrix:
\begin{equation}\label{Pauli}
	\sigma^{(0)}=\begin{pmatrix}1&0\\0&1\end{pmatrix}\qquad
	\sigma^{(1)}=\begin{pmatrix}0&1\\1&0\end{pmatrix}\qquad
	\sigma^{(2)}=\begin{pmatrix}0&-\im\\\im&0\end{pmatrix}\qquad
	\sigma^{(3)}=\begin{pmatrix}1&0\\0&-1\end{pmatrix}
\end{equation}
The labelling is cyclic so that $\sigma_{n+1}^{(a)}\equiv\sigma_{1}^{(a)}$.  These matrices act on the Hilbert space of $n$ distinguishable qubits, $\left(\mathbb{C}^2\right)^{\otimes n}$, which is the $n$ fold tensor product of the individual qubit Hilbert spaces $\mathbb{C}^2$.

This class is seen to contain the most general Hamiltonians, up to the addition of the identity operator, that describe a ring of qubits which are only able to interact with their nearest-neighbours.  It characterises the Hamiltonians studied within a random matrix theory framework in our related paper \cite{KLW2014_2}.

The second class is a subclass of the first.  It includes only those Hamiltonians with a translational symmetry along the chain, specifically the matrices
\begin{equation}
  H_n^{(inv)}=\frac{1}{\sqrt{n}}\sum_{j=1}^n\sum_{a=0}^3\sum_{b=1}^3\alpha_{a,b}\sigma_j^{(a)}\sigma_{j+1}^{(b)}
\end{equation}
for any $\alpha_{a,b}=\alpha_{a,b}(n)\in\mathbb{R}$.

$H_n^{(inv)}$ is invariant under conjugation by the unitary translation operator, $T$, which acts as
\begin{equation}\label{T}
  T\Big(|x_1\rangle\otimes|x_2\rangle\otimes\dots\otimes|x_n\rangle\Big)=|x_n\rangle\otimes|x_1\rangle\otimes\dots\otimes|x_{n-1}\rangle
\end{equation}
if $|0\rangle$ and $|1\rangle$ are an eigenbasis for $\sigma^{(3)}$, for all $x_1,x_2,\dots,x_n\in\{0,1\}$.

\section{Results}\label{results}
The first analytic result of this paper concerns entanglement in the eigenstates of general translationally invariant Hamiltonians of a ring of $n$ sequentially labelled qubits.  Let such Hamiltonians be denoted $G_n^{\lab}$ so that $G_n^{\lab}=TG_n^{\lab}T^{-1}$.  In particular this includes the nearest-neighbour translationally invariant Hamiltonians $H_n^{(inv)}$ defined previously.  Denoting the consecutive qubits labelled $(1)$ to $(l)$ by $\mathcal{A}$ and the qubits labelled $(l+1)$ to $(n)$ by $\mathcal{B}$, see Figure \ref{split}, let $\rho_k=|\psi_k\rangle\langle\psi_k|$, for a joint eigenstate $|\psi_k\rangle$ of $G_n^{(inv)}$ and the translation operator $T$, and let the associated reduced density matrix on $\mathcal{A}$ be $\rho_{k,\mathcal{A}}=\Tr_\mathcal{B}\left(\rho_k\right)$.  The purity $\Tr_{\mathcal{A}}\left(\rho_{k,\mathcal{A}}^2\right)$, of $\rho_{k,\mathcal{A}}$, will be used as an indicator of the entanglement present in the joint eigenstates of $T$ and $G_n^{\lab}$, $|\psi_k\rangle$, between $\mathcal{A}$ and $\mathcal{B}$.  Here the maximal value of $1$ corresponds to a product state across $\mathcal{A}$ and $\mathcal{B}$ and the minimal value of $\frac{1}{2^l}$ corresponds to a maximally entangled state across $\mathcal{A}$ and $\mathcal{B}$.

\begin{figure}
  \centering
  \includegraphics{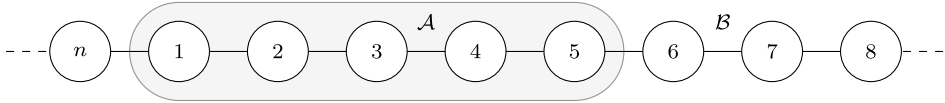}
  \caption{A nearest-neighbour spin chain of $n$ qubits.  The circles represent the qubits labelled $1$ to $n$ and the links the nearest-neighbour interactions.  The system is split into two subsystems; subsystem $\mathcal{A}$ comprised of the $l=5$ qubits labelled $1$ to $5$ (shading) and subsystem $\mathcal{B}$ the remaining $n-5$ qubits.}
  \label{split}
\end{figure}

The average of the purity $\Tr_{\mathcal{A}}\left(\rho_{\mathcal{A},k}^2\right)$ over all eigenstates for any fixed Hamiltonian $G_n^{(inv)}$ (in particular for a nearest-neighbour Hamiltonian $H_n^{(inv)}$) is bounded by:
\begin{theorem}[Average eigenstate purity on $l$ qubits for general translationally invariant Hamiltonians]\label{th2}
For $n=1,2,\dots$ let
\begin{equation}
  G_n^{(inv)}=TG_n^{\lab}T^{-1}
  =\sum_{k=1}^{2^n}\lambda_k|\psi_k\rangle\langle\psi_k|
\end{equation}
be a fixed sequence of any translationally invariant spin chain Hamiltonians of a ring of $n$ sequentially labelled qubits, where $\{|\psi_k\rangle\,:\,k=1,\dots,2^n\}$ is a common orthonormal basis of $G_n^{\lab}$ and the translation operator $T$, with the corresponding eigenvalues $\lambda_1\leq\lambda_2\leq\dots\leq\lambda_{2^n}$ of $G_n^{\lab}$.  Then, for $2l<n$, the reduced density matrices $\rho_{k,\mathcal{A}}=\Tr_\mathcal{B}\left(|\psi_k\rangle\langle\psi_k|\right)$ satisfy,
\begin{equation}
  \frac{1}{2^l}\leq\frac{1}{2^n}\sum_{k=1}^{2^n}\Tr_{\mathcal{A}}\left(\rho_{k,\mathcal{A}}^2\right)\leq\frac{1}{2^l}+\frac{2^l}{n}
\end{equation}
where $\mathcal{A}$ is the Hilbert space of the qubits labelled $(1)$ to $(l)$ and $\mathcal{B}$ is the Hilbert space of the remaining qubits.
\end{theorem}

As the purity of $\rho_{k,\mathcal{A}}$ is at least $\frac{1}{2^l}$ the next corollary follows immediately:
\begin{corollary}[Minimal purity for almost all eigenstates]\label{epsilon}
For any fixed $\epsilon>0$ the proportion of eigenstates $|\psi_k\rangle$ for which $\Tr_{\mathcal{A}}\left(\rho_{k,\mathcal{A}}^2\right)\geq\frac{1}{2^l}+\epsilon$ tends to zero as $n\to\infty$ for fixed $l$.
\end{corollary}

Theorem \ref{th2} and Corollary \ref{epsilon} do not imply that the bound $\frac{1}{2^l}\leq\Tr_{\mathcal{A}}\left(\rho_{k,\mathcal{A}}^2\right)\leq\frac{1}{2^l}+\frac{2^l}{n}$ holds for every eignestate $|\psi_k\rangle$ but rather `on average' over the entire eigenbasis.  However in the large $n$ limit almost all (in the sense of Corollary \ref{epsilon}) eigenstates satisfy $\frac{1}{2^l}\leq\Tr_{\mathcal{A}}\left(\rho_{k,\mathcal{A}}^2\right)<\frac{1}{2^l}+\epsilon$ for any fixed $\epsilon>0$, that is have a purity arbitrarily close to a value corresponding to maximal bipartite entanglement.

The proof of Theorem \ref{th2} is given in Section \ref{Th2proof}.  Here, it is also shown that many matrices of the form $H_n^{\lab}$ have a non-degenerate spectrum (although this is not a condition of Theorem \ref{th2}), so that Theorem \ref{th2} describes the unique (up to phase) eigenstates of $H_n^{(inv)}$ in these cases.  A slight refinement of the theorem in the case of some matrices of the form $H_n$ is also given for the case $l=1$.

The second analytic result of this paper concerns the convergence of the density of states probability measure for many fixed sequences of matrices $H_n$ as $n\to\infty$ (which includes the translationally invariant Hamiltonians as a special case).  The density of states probability measure determines the density of the eigenvalues of $H_n$ on the real line.  The theorem provides a contrast to the eigenstate statistics of translationally invariant Hamiltonians, which are that expected from systems of highly interacting qubits.

\begin{theorem}[Limiting density of states measure for a sequence of Hamiltonians]\label{DOS}
	For $n=2,3,\dots$ let
	\begin{equation}
		H_n=\frac{1}{\sqrt{n}}\sum_{j=1}^n\sum_{a=0}^3\sum_{b=1}^3\alpha_{a,b,j}\sigma_j^{(a)}\sigma_{j+1}^{(b)}
	\end{equation}
	be a fixed sequence of spin chain Hamiltonians for any $\alpha_{a,b,j}=\alpha_{a,b,j}(n)\in\mathbb{R}$ such that for each $n$
	\begin{equation}\label{range}
		\frac{1}{2^n}\Tr_{\mathcal{AB}}\left(H_n^2\right)\equiv\frac{1}{n}\sum_{j=1}^n\sum_{a=0}^3\sum_{b=1}^3\alpha_{a,b,j}^2=1\qquad\text{and}\qquad
		\left|\alpha_{a,b,j}\right|<C
	\end{equation}
	for some positive constant $C$ independent of $a$, $b$, $j$ and $n$.  Then the density of states measure
	\begin{equation}
		\rho_n(\lambda)\di\lambda=\frac{1}{2^n}\sum_{k=1}^{2^n}\delta(\lambda-\lambda_k)\di\lambda
	\end{equation}
	for the eigenenergies  $\lambda_1,\dots,\lambda_{2^n}$ of $H_n$ tends weakly to that of a standard normal distribution, that is
	\begin{equation}
		\lim_{n\to\infty}\int_{-\infty}^{x^+}\rho_n(\lambda)\di\lambda=\frac{1}{\sqrt{2\pi}}\int_{-\infty}^x\e^{-\frac{\lambda^2}{2}}\di\lambda
	\end{equation}
        for all $x\in\mathbb{R}$ (the notation $x^+$ represents the limit as $x$ is approached from above).
\end{theorem}

It is noted that this theorem holds, in particular, for translationally invariant Hamiltonians.

The proof of this theorem follows closely a calculation in \cite{Mahler,Mahler2005}. There it was shown that for many Hamiltonians $H_n$ and product states $|\phi\rangle$, over the $n$ qubits, that $\Tr_{\mathcal{AB}}\left(\e^{\im tH_n}|\phi\rangle\langle\phi|\right)$ tends to the characteristic function of some scaled and shifted Gaussian random variable.  A modified version of this calculation is given in Section \ref{Th1proof} to prove Theorem \ref{DOS}.  This modification leads to slightly more general constraints within the proof which carry forward to the generalisation to more elaborate interaction geometries, given thereafter.  This modified calculation also highlights the similarity to the related random matrix calculation, which this extends, given in our previous paper \cite{KLW2014_2}.

\section{Proof of Theorem \ref{th2} and extended results}\label{Th2proof}
The proof of Theorem \ref{th2} will now be given:
\begin{proof}
The proof is in two parts.  First the Hermitian matrix $\rho_{k,\mathcal{A}}$ will be expanded over an orthonormal Pauli matrix basis.  The non-identity coefficients in this expansion will then be bounded using the translational invariance of $G_n^{\lab}$, so that the desired trace can be calculated.

The space of $2^l\times2^l$ Hermitian matrices admits the orthonormal basis
\begin{equation}
  \left\{\sigma^{( a_1 )}\otimes\sigma^{( a_2 )}\otimes\dots\otimes\sigma^{( a_l )}\,:\,a_1,a_2,\dots,a_l=0,1,2,3\right\}
\end{equation}
with respect to the (scaled) Hilbert-Schmidt inner product $(A,B)=\frac{1}{2^l}\Tr_{\mathcal{A}}(AB^\dagger)$.  As such, the reduced density matrices $\rho_{k,\mathcal{A}}$ have the decomposition
\begin{equation}
  \rho_{k,\mathcal{A}}=\sum_{a_1,\dots,a_l=0}^3\frac{1}{2^l}\Tr_{\mathcal{A}}\Big(\sigma^{( a_1 )}\otimes\dots\otimes\sigma^{( a_l )}\Tr_\mathcal{B}(|\psi_k\rangle\langle\psi_k|)\Big)\sigma^{( a_1 )}\otimes\dots\otimes\sigma^{( a_l )}
\end{equation}
Combining the trace over the first $l$ qubits and the trace over the last $n-l$ qubits reduces this to
\begin{equation}
  \sum_{a_1,\dots,a_l=0}^3\frac{1}{2^l}\Tr_{\mathcal{AB}}\Big(\sigma_{  1  }^{( a_1 )}\dots\sigma_{  l  }^{( a_l )}|\psi_k\rangle\langle\psi_k|\Big)\sigma^{( a_1 )}\otimes\dots\otimes\sigma^{( a_l )}
\end{equation}
or equivalently
\begin{equation}\label{exp}
  \rho_{k,\mathcal{A}}=\frac{1}{2^l}\sum_{a_1,\dots,a_l=0}^3\langle\psi_k|\sigma_{  1  }^{( a_1 )}\dots\sigma_{  l  }^{( a_l )}|\psi_k\rangle\sigma^{( a_1 )}\otimes\dots\otimes\sigma^{( a_l )}
\end{equation}

It now remains to bound the coefficients $\langle\psi_k|\sigma_{  1  }^{( a_1 )}\dots\sigma_{  l  }^{( a_l )}|\psi_k\rangle$.  By construction, $T$ is unitary and $T|\psi_k\rangle=\e^{\im\theta_k}|\psi_k\rangle$ for some $\theta_k\in[0,2\pi)$.  Therefore,
\begin{align}\label{M}
  \langle\psi_k|\sigma_{  1  }^{( a_1 )}\dots\sigma_{  l  }^{( a_l )}|\psi_k\rangle
  &=\frac{1}{n}\sum_{j=0}^{n-1}\langle\psi_k|T^{j}\sigma_{  1  }^{( a_1 )}\dots\sigma_{  l  }^{( a_l )}T^{-j}|\psi_k\rangle=\frac{1}{\sqrt{n}}\langle\psi_k|M(\boldsymbol{a})|\psi_k\rangle
\end{align}
where $M(\boldsymbol{a})$, for $\boldsymbol{a}=(a_1,\dots,a_l)$, is the Hermitian matrix on $n$ qubits
\begin{equation}
  M(\boldsymbol{a})=\frac{1}{\sqrt{n}}\sum_{j=0}^{n-1}T^{j}\sigma_{  1  }^{( a_1 )}\dots\sigma_{  l  }^{( a_l )}T^{-j}=\frac{1}{\sqrt{n}}\sum_{j=0}^{n-1}\sigma_{  1+j  }^{( a_1 )}\dots\sigma_{  l+j  }^{( a_l )}
\end{equation}
Excluding the case where $\boldsymbol{a}=\boldsymbol{0}$, if $2l<n$
\begin{equation}\label{bound2}
  \Tr_{\mathcal{AB}}(M(\boldsymbol{a})M(\boldsymbol{a})^\dagger)=\Tr_{\mathcal{AB}}\left(M(\boldsymbol{a})^2\right)=\frac{1}{n}\Tr_{\mathcal{AB}}\left(\left(\sum_{j=0}^{n-1}\sigma_{  1+j  }^{( a_1 )}\dots\sigma_{  l+j  }^{( a_l )}\right)^2\right)=\frac{1}{n}\Tr_{\mathcal{AB}}\left(\sum_{j=0}^{n-1}I_{2^n}\right)=2^n
\end{equation}
as all the off-diagonal terms, in the square above, have zero trace by the orthonormality property of the Pauli matrix basis.  It is important to note that all the operators $\sigma_{  1+j  }^{( a_1 )}\dots\sigma_{  l+j  }^{( a_l )}$ for $j=0,\dots,n-1$ within the square above are unique as $2l<n$ and $\boldsymbol{a}\neq\boldsymbol{0}$.  This follows as the operators are only supported on $l$ consecutive qubits.

We also note that later in the paper prime values of $n$ will be required in order to determine the simplicity of certain eigenvalues.  At this point though neither prime values of $n$ nor eigenvalue simplicity are required.

These facts may now be combined to complete the proof.  By equations (\ref{exp}) and (\ref{M}) and the orthonormality property of the Pauli matrix basis,
\begin{equation}
  \Tr_{\mathcal{A}}\left(\rho_{k,\mathcal{A}}^2\right)=\sum_{a_1,\dots,a_l=0}^3\frac{1}{2^l}\langle\psi_k|\sigma_{  1  }^{( a_1 )}\dots\sigma_{  l  }^{( a_l )}|\psi_k\rangle^2=\sum_{a_1,\dots,a_l=0}^3\frac{1}{n2^l}\langle\psi_k|M(\boldsymbol{a})|\psi_k\rangle^2
\end{equation}
Since $\sum_k\langle\psi_k|M(\boldsymbol{a})|\psi_k\rangle^2\leq\Tr_{\mathcal{AB}}(M(\boldsymbol{a})M(\boldsymbol{a})^\dagger)=\Tr_{\mathcal{AB}}\left(M(\boldsymbol{a})^2\right)$, it now follows from equation (\ref{bound2}) that
\begin{align}
  \sum_{k=1}^{2^n}\Tr_{\mathcal{A}}\left(\rho_{k,\mathcal{A}}^2\right)
  &=\sum_{k=1}^{2^n}\frac{1}{2^l}\langle\psi_k|\psi_k\rangle^2+\sum_{k=1}^{2^n}\sum_{\genfrac{}{}{0pt}{}{a_1,\dots,a_l=0}{\text{not all zero}}}^3\frac{1}{n2^l}\langle\psi_k|M(\boldsymbol{a})|\psi_k\rangle^2\leq\frac{2^n}{2^l}+\frac{2^l2^n}{n}
\end{align}
Dividing through by $2^n$ and recalling that $\Tr_{\mathcal{A}}\left(\rho_{k,\mathcal{A}}^2\right)\geq\frac{1}{2^l}$ for each $k$, as $\rho_{k,\mathcal{A}}$ is a $2^l$ dimensional density matrix, completes the proof.
\end{proof}

\subsection{Extension to qudits and higher dimensional lattices}
This proof relies on the fact that an orthogonal basis $\{B_i\}$ for Hermitian operators supported on subsystem $\mathcal{A}$ can be translated along the chain many times by a translation operator $T$ with
\begin{equation}
 \langle\psi_k|B_i|\psi_k\rangle=\langle\psi_k|T^{j}B_iT^{-j}|\psi_k\rangle
\end{equation}
for all $j=0,\dots,n-1$, as used in equation (\ref{M}).  This situation is not restricted to a one dimensional system of qubits.  A similar result will hold for qudits or higher dimensional systems, so long as there exists a suitable translation symmetry in the system so that an analogous identity to (\ref{M}) holds.

\subsection{Eigenstate uniqueness}\label{uniqueness}
The eigenstates of matrices with a non-degenerate spectrum are unique up to phase.  In this case Theorem \ref{th2} becomes a theorem describing the unique (up to phase) eigenstates of $G_n^{\lab}$.  The following two lemmas shows that this is the case for most matrices $H_n^{\lab}$ (as particular examples of Hamiltonians $G_n^{\lab}$ which only include nearest-neighbour interactions) when $n$ is an odd prime:

\begin{lemma}\label{dist}
	For odd prime values of $n$, there exists some $\epsilon\in\mathbb{R}$ such that the matrix
	\begin{equation}
		H_{n}^{(\epsilon XY+Z)}=\sum_{j=1}^n\left(\epsilon\sigma_{  j  }^{(1)}\sigma_{  j+1  }^{(2)}+\sigma_{  j  }^{(3)}\right)
	\end{equation}
	has a non-degenerate spectrum.
\end{lemma}
\begin{proof}
Let $n$ be an odd prime.  The $2^n$ eigenvalues of $H_{n}^{(\epsilon XY+Z)}$ are given by
\begin{equation}
	\lambda_{\boldsymbol{x}}=\sum_{j=1}^n(2x_j-1)\left(\epsilon\mu_j-\sqrt{\epsilon^2\mu_j^2+1}\right),\qquad\qquad\mu_j=\sin\left(\frac{2\pi j}{n}\right)
\end{equation}
for the multi-index $\boldsymbol{x}=(x_1,\dots,x_n)\in\{0,1\}^n$, as seen in Appendix \ref{JW}.

First, it will be shown that for odd prime values of $n$ the values $\{\mu_j\}_{j=1}^{\frac{n-1}{2}}$ are linearly independent over the integers.  With $\omega=\e^{\frac{2\pi\im}{n}}$,
\begin{equation}
	\mu_j=\sin\left(\frac{2\pi j}{n}\right)=\frac{\omega^j-\omega^{-j}}{2\im}
\end{equation}
so that for integers $a_j$
\begin{align}
	\sum_{j=1}^{\frac{n-1}{2}}a_j\mu_j
	&=\frac{1}{2\im}\sum_{j=1}^{\frac{n-1}{2}}a_j\left(\omega^j-\omega^{-j}\right)
	=\frac{1}{2\im}\sum_{j=1}^{n-1}b_j\omega^j
\end{align}
where
\begin{equation}
	b_j=
	\begin{cases}
		a_j\qquad\quad&\text{if }j=1,\dots,\frac{n-1}{2}\\
		-a_{n-j}&\text{if }j=\frac{n+1}{2},\dots,n-1
	\end{cases}
\end{equation}
The non-zero powers of $\omega$ are linearly independent over the integers \cite[Lemma 2.11]{Lenstra1979}; hence
\begin{equation}
	\sum_{j=1}^{n-1}b_j\omega^j=0\quad\iff\quad b_j=0\,\,\forall j\quad\iff\quad a_j=0\,\,\forall j
\end{equation}
from which it is concluded that the $\{\mu_j\}_{j=1}^{\frac{n-1}{2}}$ are linearly independent over the integers.

For small $\epsilon$, the eigenvalues $\lambda_{\boldsymbol{x}}(\epsilon)$ admit the expansion
\begin{equation}
	\lambda_{\boldsymbol{x}}(\epsilon)=\sum_{j=1}^n(2x_j-1)\left(-1+\epsilon\mu_j-\epsilon^2\frac{\mu_j^2}{2}+O\left(\epsilon^4\right)\right)
\end{equation}
Suppose, for a contradiction, that two eigenvalues are equal in some neighbourhood of $\epsilon=0$, that is for $\boldsymbol{x}\neq\boldsymbol{y}$,
\begin{equation}
	0=\lambda_{\boldsymbol{x}}(\epsilon)-\lambda_{\boldsymbol{y}}(\epsilon)
	=-2\sum_{j=1}^n(x_j-y_j)+2\epsilon\sum_{j=1}^n\mu_j(x_j-y_j)-\epsilon^2\sum_{j=1}^n\mu_j^2(x_j-y_j)+O\left(\epsilon^4\right)
\end{equation}
Comparing the $\epsilon^0$ coefficient (and setting $d_j=x_j-y_j$) gives
\begin{equation}
	0=\sum_{j=1}^nd_j
\end{equation}
Comparing the $\epsilon^1$ coefficient gives
\begin{equation}
	0=\sum_{j=1}^n\mu_jd_j=\sum_{j=1}^{\frac{n-1}{2}}\sin\left(\frac{2\pi j}{n}\right)\left(d_j-d_{n-j}\right)
\end{equation}
as $\mu_n=0$ and $\mu_j=-\mu_{n-j}$ for $j=1,\dots,\frac{n-1}{2}$.  By the linear independence of the $\{\mu_j\}_{j=1}^{\frac{n-1}{2}}$ over the integers, this implies that $d_j=d_{n-j}$ for all $j=1,\dots,\frac{n-1}{2}$.  In particular, substituting this into the $\epsilon^0$ result gives
\begin{equation}
  0=2\sum_{j=1}^{\frac{n-1}{2}}d_j+d_n
\end{equation}
from which, since the first term in even and the second term is either $-1$, $0$ or $1$, implies that $d_n=0$.
Comparing the $\epsilon^2$ coefficient gives
\begin{equation}
  0=\sum_{j=1}^n\mu_j^2d_j=\sum_{j=1}^{n-1}\mu_j^2d_j=\frac{1}{(2\im)^2}\sum_{j=1}^{n-1}(\omega^j-\omega^{-j})^2d_j
  =-\frac{1}{4}\sum_{j=1}^{n-1}(\omega^{2j}+\omega^{-{2j}})d_j+\frac{1}{4}\sum_{j=1}^{n-1}2d_j
\end{equation}
It has already been seen that the second term in this last expression is zero from the $\epsilon^0$ result, therefore
\begin{equation}
  0=-\frac{1}{4}\sum_{j=1}^{n-1}(\omega^{2j}+\omega^{-{2j}})d_j=-\frac{1}{4}\sum_{j=1}^{n-1}\omega^{2j}\left(d_j+d_{n-j}\right)
\end{equation}
(by transforming $j\to n-j$ in $\omega^{-2j}d_j$).  Now by the linear independence of the non-zero powers of $\omega$ over integers, $d_j=-d_{n-j}$ for all $j=1,\dots,n-1$.  It has already been seen that $d_n=0$ and that $d_j=d_{n-j}$ for all $j=1,\dots,\frac{n-1}{2}$ so that it is concluded that $d_j=0$ for all $j$.  That is $\boldsymbol{x}=\boldsymbol{y}$, a contradiction.  Therefore there must exist some $\epsilon$ for which $H_n^{(\epsilon XY+Z)}$ has a non-degenerate spectrum by the analyticity of $\lambda_{\boldsymbol{x}}(\epsilon)$ for $\epsilon>0$.
\end{proof}

\begin{lemma}\label{exist}
  For odd prime values of $n$ the matrix
  \begin{equation}
    H_n^{\lab}(\boldsymbol{\alpha})=\frac{1}{\sqrt{n}}\sum_{j=1}^n\sum_{a=0}^3\sum_{b=1}^3\alpha_{a,b}\sigma_j^{(a)}\sigma_{j+1}^{(b)}
  \end{equation}
  for $\boldsymbol{\alpha}=(\alpha_{0,1},\dots,\alpha_{3,3})\in\mathbb{R}^{12}$, generically has a non-degenerate spectrum.
\end{lemma}
\begin{proof}
  Let $V$ be the $2^n\times2^n$ Vandermonde matrix with elements $V_{j,k}=\lambda_j^{k-1}$ for the eigenvalues $\lambda_1,\dots,\lambda_{2^n}$ of $H_n^{\lab}$.  Then by the properties of the Vandermonde matrix
  \begin{equation}\label{eq1}
    \det_{1\leq j,k\leq2^n}\left(V^\dagger V\right)=\prod_{1\leq j<k\leq2^n}\left(\lambda_k-\lambda_j\right)^2
  \end{equation}
  and also, by direct calculation,
  \begin{equation}\label{eq2}
    \left(V^\dagger V\right)_{j,k}=\Tr\left(\left(H_n^{\lab}(\boldsymbol{\alpha})\right)^{j+k-2}\right)
  \end{equation}
  The function
  \begin{equation}
    \det_{1\leq j,k\leq2^n}\left(V^\dagger V\right)
  \end{equation}
  is then a polynomial in the elements of $\boldsymbol{\alpha}$ by (\ref{eq2}) and is zero iff $H_n^{\lab}$ has at least one repeated eigenvalue by (\ref{eq1}).  Lemma \ref{dist} shows that there exists some $\boldsymbol{\alpha}_0$ such that this polynomial is non-zero.  Hence the Lebesgue measure of the zeros of this polynomial must be zero, completing the proof.
\end{proof}

\subsection{Non-translationally invariant Hamiltonians}
An exact result for each eigenstate is possible in the $l=1$ case for some particular instances of the matrices $H_n$ with non-degenerate spectrum:
\begin{theorem}[Eigenstate purity on one qubit for some matrices $H_n$]\label{th3}
  Let
  \begin{equation}
    {H}_n^{(pair)}=\sum_{j=1}^n\sum_{a,b=1}^3\alpha_{a,b,j}\sigma_{  j  }^{(a)}\sigma_{  j+1  }^{(b)}=\sum_{k=1}^{2^n}\lambda_k|\psi_k\rangle\langle\psi_k|
  \end{equation}
  for any $\alpha_{a,b,j}\in\mathbb{R}$ and where $\{|\psi_k\rangle\,:\,k=1,\dots,2^n\}$ are eigenstates of $H_n^{(pair)}$ with the corresponding eigenvalues $\lambda_1<\lambda_2<\dots<\lambda_{2^n}$ (the absence of local terms proportional to $\sigma_j^{(a)}$ should be noted).  Then the reduced density matrices $\rho_{k,\mathcal{A}}=\Tr_\mathcal{B}\left(|\psi_k\rangle\langle\psi_k|\right)$ satisfy,
\begin{equation}
  \Tr_{\mathcal{A}}\left(\rho_{k,\mathcal{A}}^2\right)=\frac{1}{2}
\end{equation}
where $\mathcal{A}$ is the Hilbert space of the qubit labelled $(1)$ and $\mathcal{B}$ is the Hilbert space of the remaining qubits.
\end{theorem}

\begin{proof}
As in the proof of Theorem \ref{th2}, see equation (\ref{exp}), the density matrix $\rho_{k,\mathcal{A}}$ may be written as
\begin{equation}\label{exp2}
  \rho_{k,\mathcal{A}}=\frac{1}{2}\sum_{a=0}^3\langle\psi_k|\sigma_{  1  }^{(a)}|\psi_k\rangle\sigma^{(a)}
\end{equation}

For the unitary operator $S={\sigma^{(2)}}^{\otimes n}$ and $a=1,2,3$, $j=1,\dots,n$ it can be shown that $S\sigma_j^{(a)}S=-\overline{\sigma_j^{(a)}}$, where the bar denotes complex conjugation of the matrix elements the standard basis as used in (\ref{Pauli}).  Therefore $SH_n^{(pair)}S=SH_n^{(pair)}S^\dagger=\overline{H_n^{(pair)}}$ as non-identity Pauli matrices only occur in pairs in $H_n^{(pair)}$ and all the coefficients $\alpha_{a,b,j}$ in $H_n^{(pair)}$ are real.  Then for the eigenstate $|\psi_k\rangle$ of $H_n^{(pair)}$ with eigenvalue of $\lambda_k$ it follows that
\begin{equation}
  \overline{H_n^{(pair)}}\Big(S|\psi_k\rangle\Big)=S H_n^{(pair)}|\psi_k\rangle=\lambda_k\Big(S|\psi_k\rangle\Big)
\end{equation}
so that both $|\psi_k\rangle$ and $\overline{S|\psi_k\rangle}$ are eigenstates of $H_n^{(pair)}$ with eigenvalue $\lambda_k$.  By assumption $\lambda_k$ is non-degenerate so that $\overline{S|\psi_k\rangle}=\e^{\im\theta_k}|\psi_k\rangle$ for some $\theta_k\in[0,2\pi)$.  Since $\overline{S^\dagger}\sigma_1^{(a)}\overline{S}=S^\dagger\sigma_1^{(a)}S=-\overline{\sigma_1^{(a)}}$, it then follows that for $a=1,2,3$,
\begin{equation}\label{sym}
\langle\psi_k|\sigma_{  1  }^{(a)}|\psi_k\rangle=\overline{\langle\psi_k|S^\dagger}\sigma_{  1  }^{(a)}\overline{S|\psi_k\rangle}=-\overline{\langle\psi_k|\sigma_{  1  }^{(a)}|\psi_k\rangle}
\end{equation}
As $\sigma_{  1  }^{(a)}$ is Hermitian, $\overline{\langle\psi_k|\sigma_{  1  }^{(a)}|\psi_k\rangle}=\langle\psi_k|\sigma_{  1  }^{(a)}|\psi_k\rangle$, and it is concluded that $\langle\psi_k|\sigma_{  1  }^{(a)}|\psi_k\rangle=0$.  Therefore, from equation (\ref{exp2}),
\begin{equation}
  \rho_{k,\mathcal{A}}=\frac{1}{2}\langle\psi_k|\sigma_{  1  }^{(0)}|\psi_k\rangle\sigma^{(0)}=\frac{I_2}{2}
\end{equation}
The value of $\Tr_{\mathcal{A}}\left(\rho_{k,\mathcal{A}}^2\right)$ is now seen to be $\frac{1}{2}$, for all values of $k$, as claimed.
\end{proof}

In fact this proof gives a more general result.  Subsystem $\mathcal{A}$ may be extended to the qubits labelled $(1)$ to $(l)$ (for $l=1,\dots,n$).  The density matrix $\rho_{k,\mathcal{A}}$ then reads
\begin{equation}
  \rho_{k,\mathcal{A}}=\frac{1}{2^l}\sum_{a_1,\dots,a_l=0}^3\langle\psi_k|\sigma_{  1  }^{( a_1 )}\dots\sigma_{  l  }^{( a_l )}|\psi_k\rangle\sigma^{( a_1 )}\otimes\dots\otimes\sigma^{( a_l )}
\end{equation}
as in equation (\ref{exp}).  By an analogous reasoning to that in (\ref{sym}), it is then deduced that the coefficients $\langle\psi_k|\sigma_{  1  }^{( a_1 )}\dots\sigma_{  l  }^{( a_l )}|\psi_k\rangle$, for which an odd number of the $\{a_j\}_{j=1}^l$ are non-zero, vanish.

\section{Proof of Theorem \ref{DOS} and extended results}\label{Th1proof}
The proof of Theorem \ref{DOS} will now be given.  The proof of this theorem is an adaptation of that in \cite{Mahler,Mahler2005} for the distribution of eigenvalues of a generic qubit spin chain Hamiltonian, in a fixed product state.  The Hamiltonian will be split into many commuting blocks to which Lyapunov's central limit theorem \cite{Bill} applies, by removing some interaction terms.  The error generated by the removal of these terms is shown to vanish in the large $n$ limit.
\begin{proof}
  First, $H_n$ is split into blocks $b_k$ acting nontrivially on $l$ consecutive qubits, in the same fashion as in \cite{Mahler}, that is let
	\begin{align}\label{b}
		b_k&=\sum_{j=1}^{l-1}h_{(k-1)\cdot l+j}
	\end{align}
	for $k=1,2,\dots,\left\lceil\frac{n}{l}\right\rceil$, where
	\begin{equation}\label{h}
		h_0=\frac{1}{\sqrt{n}}\sum_{a=0}^3\sum_{b=1}^3\alpha_{a,b,n}\sigma_n^{(a)}\sigma_{1}^{(b)}\qquad\text{and}\qquad
		h_j=\frac{1}{\sqrt{n}}\sum_{a=0}^3\sum_{b=1}^3\alpha_{a,b,j}\sigma_j^{(a)}\sigma_{j+1}^{(b)}
	\end{equation}
	for $j=1,2,\dots,n-1$ and $h_j=0$ otherwise.  Also let
	\begin{equation}
		B=\sum_{k=1}^{\left\lceil\frac{n}{l}\right\rceil}b_k\qquad\text{and}\qquad
		L=\sum_{k=1}^{\left\lceil\frac{n}{l}\right\rceil}h_{(k-1)\cdot l}
	\end{equation}
	so that $H_n$ is a sum of the blocks $B$ and links $L$, that is $H_n=B+L$.  The value $l=l(n)$ is chosen such that
	\begin{equation}\label{l}
		\lim_{n\to\infty}\frac{1}{l}=0\qquad\text{and}\qquad\lim_{n\to\infty}\frac{l}{n}=0
	\end{equation}

	The characteristic function $\psi_n(t)$ associated to $\rho_n(\lambda)$ is, by definition,
	\begin{equation}
		\psi_n(t)=\mathbb{E}_{\rho_n(\lambda)}\left(\e^{\im t\lambda}\right)=\frac{1}{2^n}\sum_{k=1}^{2^n}\int_{-\infty}^\infty\e^{\im t\lambda}\delta(\lambda-\lambda_k)\di\lambda=\frac{1}{2^n}\Tr_{\mathcal{AB}}\left(\e^{\im tH_n}\right)
	\end{equation}
	Let the analogous characteristic function for the matrix $B$ be $\phi_n(t)=\frac{1}{2^n}\Tr_{\mathcal{AB}}\left(\e^{\im tB}\right)$.  It will now be shown that
	\begin{equation}
		\lim_{n\to\infty}\left|\psi_n(t)-\phi_n(t)\right|=\lim_{n\to\infty}\left|\frac{1}{2^n}\Tr_{\mathcal{AB}}\left(\e^{\im tH_n}\right)-\frac{1}{2^n}\Tr_{\mathcal{AB}}\left(\e^{\im tB}\right)\right|=0
	\end{equation}
	for all fixed $t\in\mathbb{R}$.  The integral identity for any $2^n\times2^n$ Hermitian matrices $X$ and $Y$ \cite{Mahler}
	\begin{equation}
		\e^{\im t(X+Y)}-\e^{\im tX}=\im\int_0^t\e^{\im(t-s)(X+Y)}Y\e^{\im tX}\di s
	\end{equation}
	along with the Cauchy-Schwartz inequality, for any $2^n\times2^n$ matrix $M$,
	\begin{equation}
  	\left|\Tr\left(M\right)\right|^2=\left|\sum_{j=1}^{2^n}M_{jj}\right|^2\leq\sum_{j=1}^{2^n}|M_{jj}|^2\sum_{k=1}^{2^n}|1|^2\leq2^n\Tr\left(MM^\dagger\right)
	\end{equation}
	and the triangle inequality yield that
	\begin{align}\label{genbound1}
		\left|\frac{1}{2^n}\Tr_{\mathcal{AB}}\left(\e^{\im tH_n}\right)-\frac{1}{2^n}\Tr_{\mathcal{AB}}\left(\e^{\im tB}\right)\right|
		&\leq\int_0^t\left|\frac{1}{2^n}\Tr_{\mathcal{AB}}\left(\e^{\im(t-s)H_n}L\e^{\im tB}\right)\right|\di s
		\leq \sqrt{\frac{t^2}{2^n}\Tr_{\mathcal{AB}}\left(LL^\dagger\right)}
	\end{align}
	The matrices $\sigma_j^{(a)}\sigma_{j+1}^{(b)}$ appearing in $L$ are orthonormal with respect to the (scaled) Hilbert-Schmidt inner product $(X,Y)=\frac{1}{2^n}\Tr_{\mathcal{AB}}\left(XY^\dagger\right)$ so therefore
	\begin{align}
		\frac{t^2}{2^n}\Tr_{\mathcal{AB}}\left(LL^\dagger\right)
		&\leq t^2\sum_{k=1}^{\left\lceil\frac{n}{l}\right\rceil}\sum_{a=0}^3\sum_{b=1}^3\frac{\alpha_{a,b,(k-1)\cdot l}^2}{n}
		\leq t^2\left\lceil\frac{n}{l}\right\rceil\frac{12C^2}{n}\to0
	\end{align}
	for all fixed $t\in\mathbb{R}$ by the bounds in (\ref{range}) and assumption in (\ref{l}).
	
	The characteristic function $\phi_n(t)$ factors into the product of $\left\lceil\frac{n}{l}\right\rceil$ characteristic functions as the $b_k$ are supported on disjoint collections of sites, and so $[b_k,b_{k^\prime}]=0$ for all $k\neq k^{\prime}$, that is,
	\begin{equation}
		\phi_n(t)
		=\frac{1}{2^n}\Tr_{\mathcal{AB}}\left(\e^{\im tB}\right)
		=\frac{1}{2^n}\Tr_{\mathcal{AB}}\left(\prod_{k=1}^{\left\lceil\frac{n}{l}\right\rceil}\e^{\im tb_k}\right)
		=\prod_{k=1}^{\left\lceil\frac{n}{l}\right\rceil}\frac{1}{2^n}\Tr_{\mathcal{AB}}\left(\e^{\im tb_k}\right)
	\end{equation}
	
Lyapunov's central limit theorem \cite{Bill} can now be applied to the characteristic function $\phi_n(t)$, which is the characteristic function of a sum of independent random variables, each with the characteristic function $\frac{1}{2^n}\Tr_{\mathcal{AB}}\left(\e^{\im tb_k}\right)$ respectively.  Lyapunov's central limit theorem states that if $\hat{x}_{n,1},\dots,\hat{x}_{n,r(n)}$ are some independent random variables (not necessarily identically distributed) each with finite mean $\mathbb{E}\left(\hat{x}_{n,j}\right)$ and variance $\mathbb{E}\left(\hat{x}_{n,j}^{2}\right)$, for each $n\in\mathbb{N}$ and some strictly increasing function $r:\mathbb{N}\to\mathbb{N}$, and if
\begin{equation}
	s_{n}^{2}=\sum_{j=1}^{r(n)}\mathbb{E}\left(\hat{x}_{n,j}^{2}\right)
\end{equation}
and the Lyapunov condition
\begin{equation}
	\lim_{n\to\infty}\frac{1}{s_{n}^{2+\delta}}\sum_{j=1}^{r(n)}\mathbb{E}\left(\left|\hat{x}_{n,j}-\mathbb{E}\left(\hat{x}_{n,j}\right)\right|^{2+\delta}\right)=0
\end{equation}
is satisfied for some $\delta>0$ then the distribution of the sum
\begin{equation}
	\frac{1}{s_{n}}\sum_{j=1}^{r(n)}\Big(\hat{x}_{n,j}-\mathbb{E}\left(\hat{x}_{n,j}\right)\Big)
\end{equation}
converges in distribution to a standard normal random variable.

A sufficient Lyapunov condition on the fourth moment of the distributions associated to the characteristic functions $\frac{1}{2^n}\Tr_{\mathcal{AB}}\left(\e^{\im tb_k}\right)$ reads
	\begin{equation}\label{Lya}
		\lim_{n\to\infty}\frac{1}{s_n^4}\sum_{k=1}^{\left\lceil\frac{n}{l}\right\rceil}\frac{1}{2^n}\Tr_{\mathcal{AB}}\left(b_k^4\right)=0,\qquad\qquad
		s_n^2=\sum_{k=1}^{\left\lceil\frac{n}{l}\right\rceil}\frac{1}{2^n}\Tr_{\mathcal{AB}}\left(b_k^2\right)
	\end{equation}
  Here, the first, second and fourth moments of the distributions associated with characteristic functions $\frac{1}{2^n}\Tr_{\mathcal{AB}}\left(\e^{\im tb_k}\right)$, used in the condition, are read off as the coefficients of $\frac{\im t}{1!}$, $\frac{(\im t)^2}{2!}$ and $\frac{(\im t)^4}{4!}$ respectively in the Taylor expansion of $\frac{1}{2^n}\Tr_{\mathcal{AB}}\left(\e^{\im tb_k}\right)$ about $t=0$.
  
  Before verifying that the Lyapunov condition (\ref{Lya}) holds, it is noted that as the matrices $\sigma_j^{(a)}\sigma_{j+1}^{(b)}$ appearing in $H_n$ and $b_k$ are orthonormal with respect to the (scaled) Hilbert-Schmidt inner product $(X,Y)=\frac{1}{2^n}\Tr_{\mathcal{AB}}\left(XY^\dagger\right)$,
  \begin{align}\label{genbound2}
  	\left|\frac{1}{2^n}\Tr_{\mathcal{AB}}\left(H_n^2\right)-\sum_{k=1}^{\left\lceil\frac{n}{l}\right\rceil}\frac{1}{2^n}\Tr_{\mathcal{AB}}\left(b_k^2\right)\right|
	  &=\left|\sum_{k=1}^{\left\lceil\frac{n}{l}\right\rceil}\sum_{a=0}^3\sum_{b=1}^3\frac{\alpha_{a,b,(k-1)\cdot l}^2}{n}\right|
	  \leq\left\lceil\frac{n}{l}\right\rceil\frac{12C^2}{n}
	  \to0
  \end{align}
  seen by directly calculating the traces.  Therefore $s_n^2\to1$ by the assumption that $\frac{1}{2^n}\Tr_{\mathcal{AB}}\left(H_n^2\right)=1$ in (\ref{range}).
  
  If the Lyapunov condition (\ref{Lya}) holds then Lyapunov's central limit theorem states that the distributions associated to the characteristic functions $\frac{1}{s_n}\phi_n(t)$ tend weakly to a standard normal distribution. By the continuity theorem \cite{Bill} pointwise convergence of the characteristic functions is equivalent to the weak convergence of the related distributions.  Therefore as $s_n\to1$ and $|\phi_n(t)-\psi_n(t)|\to0$ the density of states distribution $\rho_n(\lambda)$ tends weakly to that of a standard normal distribution as claimed.
  
  All that remains to be shown is that the Lyapunov condition (\ref{Lya}) holds.  This is equivalent to showing that
  \begin{equation}
  	\lim_{n\to\infty}\sum_{k=1}^{\left\lceil\frac{n}{l}\right\rceil}\frac{1}{2^n}\Tr_{\mathcal{AB}}\left(b_k^4\right)=0
  \end{equation}
  as $s_n\to1$.  By the definition of $b_k$ in (\ref{b})
  \begin{equation}
  	\frac{1}{2^n}\Tr_{\mathcal{AB}}\left(b_k^4\right)=\sum_{q,p,r,s=1}^{l-1}\frac{1}{2^n}\Tr_{\mathcal{AB}}\left(h_{(k-1)\cdot l+q}h_{(k-1)\cdot l+p}h_{(k-1)\cdot l+r}h_{(k-1)\cdot l+s}\right)
	\end{equation}
  For a term in this sum to be non-zero, it must be the case that the four factors $h_j$ can be split into two pairs with the indices of the $h_j$ differing by at most one within each pair.  If this is not the case, then the term evaluates to zero as the factors $h_j$ will act on completely separate sites and each has trace zero.  There are then at most $3^3(l-1)^2$ non-zero terms.  Each non-zero term can then by definition (\ref{h}) be expanded as ${12}^4$ terms, each containing a four fold product of the matrices $\sigma_j^{(a)}\sigma_{j+1}^{(b)}$, and each with modulus bounded by at most $\left(\frac{C}{\sqrt{n}}\right)^4$ by the assumptions in (\ref{range}) and the normalisation of the Pauli matrices.  Therefore
	\begin{equation}\label{genbound3}
		\left|\sum_{k=1}^{\left\lceil\frac{n}{l}\right\rceil}\frac{1}{2^n}\Tr_{\mathcal{AB}}\left(b_k^4\right)\right|
		\leq \left\lceil\frac{n}{l}\right\rceil3^3(l-1)^2{12}^4\left(\frac{C}{\sqrt{n}}\right)^4
		\leq3^74^4C^4\left(\frac{l}{n}+\frac{l^2}{n^2}\right)
	\end{equation}
  (seen by bounding $\left\lceil\frac{n}{l}\right\rceil$ by $\frac{n}{l}+1$ and $l-1$ by $l$) which tends to zero as $n\to\infty$ by the assumptions (\ref{l}), completing the proof.
\end{proof}

It is noted that this theorem holds, in particular, for translationally invariant Hamiltonians.

\subsection{More general interactions}\label{geometry}
More general qubit interactions can be dealt with (cf. \cite{Mahler,Mahler2005}).  Consider a sequence of systems of $n$ qubits where in each system each qubit is allowed to interact with a fixed number (independent of $n$) other qubits, for example interactions on a two dimensional lattice.  The general Hamiltonian of this system with $n$ qubits will have the form
\begin{equation}
  H_n^{(gen)}=\frac{1}{\sqrt{n}}\left(\sum_{(j, k)}\sum_{a,b=1}^3\alpha_{a,b,j,k}\sigma_j^{(a)}\sigma_k^{(b)}+\sum_{j=1}^n\sum_{a=1}^3\alpha_{a,0,j,0}\sigma_j^{(a)}\right)
\end{equation}
where the sum over $(j,k)$ represents the sum over all sites labelled $j$ and $k$ (for $j<k$) between which an interaction is present and the $\alpha_{a,b,j,k}$ are some real coefficients.  Here the local terms, proportional to $\sigma_j^{(a)}$, have been separated from the interaction terms, proportional to $\sigma_j^{(a)}\sigma_k^{(b)}$, to avoid repetitions in the sum.

For Hamiltonians $H_n^{(gen)}$ such that $r$ `links' or `interactions' of the form
\begin{equation}
  \sum_{a,b=1}^3\alpha_{a,b,j,k}\sigma_j^{(a)}\sigma_k^{(b)}
\end{equation}
can be removed to leave a sum of $m$ operators, each supported on non-intersecting subsets of the $n$ qubits each containing at most $q$ qubits, a modified version of the proof of Theorem \ref{DOS} holds.  If the $\alpha_{a,b,j,k}$ in such a $H_n^{(gen)}$ also satisfy an analogous condition to (\ref{range}), then for an analogous version of the proof of Theorem \ref{DOS} to hold it must be the case that:
\begin{equation}
  \frac{r}{n}\to0
\end{equation}
as $n\to\infty$ for the analogous equations to (\ref{genbound1}) and (\ref{genbound2}) to hold and that
\begin{equation}
  \frac{mq^2}{n^2}\to0
\end{equation}
as $n\to\infty$ for the analogous equation to (\ref{genbound3}) to hold.

For example, consider the case of a two dimensional $p\times p$ cyclic lattice of $n=p^2$ qubits.  A suitable partitioning into blocks would be to group the qubits into $m=\left\lceil\frac{p}{l}\right\rceil^2$ neighbouring $l\times l$ blocks of at most $q=l^2$ qubits (with possibly smaller blocks on the lattice's boundary) and remove all the $r=2p\left\lceil\frac{p}{l}\right\rceil$ `links' between these blocks.  The two condition above then read
\begin{equation}
  \frac{r}{n}=\frac{2p\left\lceil\frac{p}{l}\right\rceil}{p^2}\to0\qquad\text{and}\qquad
  \frac{mq^2}{n^2}=\frac{\left\lceil\frac{p}{l}\right\rceil^2l^4}{p^4}\to0
\end{equation}
Choosing $l(p)$ such that $\frac{1}{l}\to0$ and $\frac{l}{p}\to0$ as $p\to\infty$ (or equivalently $n\to\infty$) satisfies these conditions.  This then provides an example of a generalisation of the applicability of Theorem \ref{DOS}.

\subsection{Interacting qudits}\label{qudits}
Qudits (of fixed dimension $d$, independent of $n$) may also be used in place of qubits.  Given a matrix basis for each qudit site which is orthogonal under the Hilbert-Schmidt inner product, an analogous proof to that of Theorem \ref{DOS} again holds.  In effect this involves increasing the range of the indices $a$ and $b$ by some fixed amount in the proof.   This only leads to a change of the constants in the bounds given.

\subsection{A conjecture of Atas and Bogomolny}
Atas and Bogomolny have conjectured \cite{Atas2014} that for
\begin{equation}
  H_n^{(BA)}=\frac{1}{\sqrt{n}}\sum_{j=1}^n\left(\sigma_j^{(1)}\sigma_{j+1}^{(1)}+\alpha_1\sigma_j^{(1)}+\alpha_3\sigma_j^{(3)}\right)
\end{equation}
where $\alpha_1$ and $\alpha_2$ are fixed real coefficients, it is the case that for all positive integers $k$,
\begin{equation}
  \lim_{n\to\infty}\frac{1}{2^n}\Tr\left({H_n^{(BA)}}^{2k}\right)=(1+\alpha_1^2+\alpha_2^2)^{2k}\frac{(2k)!}{2^kk!}
\end{equation}

This can be seen as a corollary of Theorem \ref{DOS}.  Theorem \ref{DOS} states that for the sequence of Hamiltonians
\begin{equation}
  \frac{H_n^{(BA)}}{\sqrt{1+\alpha_1^2+\alpha_3^2}}
\end{equation}
the associated density of states distributions tend weakly to that of a standard normal distribution.  Rescaling with the constant $\sqrt{1+\alpha_1^2+\alpha_3^2}$ implies that for the sequence of Hamiltonians $H_n^{(BA)}$ the density of states distributions tend weakly to that of a  normal distribution with mean zero and variance of $1+\alpha_1^2+\alpha_3^2$.  The $2k^{th}$ moment of this limiting distribution is by definition given by
\begin{equation}
  \lim_{n\to\infty}\frac{1}{2^n}\Tr\left({H_n^{(BA)}}^{2k}\right)
\end{equation}
(by considering the coefficients of the characteristic functions $\frac{1}{2^n}\Tr\left(\e^{itH_n^{(BA)}}\right)$ of the associated density of states distributions) or equivalently by
\begin{equation}
  (1+\alpha_1^2+\alpha_3^2)^{2k}(2k-1)!!
\end{equation}
(by considering the moments of the normal distribution), proving the conjecture.

\section{Discussion and open questions}\label{dis}
\subsection{Tightness of the bound in Theorem \ref{th2} (Reduced eigenstate purity)}
Numerically there exist examples of Hamiltonians for which the bound in Theorem \ref{th2} appears asymptotically tight, that is linear in $\frac{1}{n}$ for $n=2,\dots,32$.

Figure \ref{etang}a shows the average value of the linear entropy (one minus the purity) of the reduced eigenstates over $s=2^6$ random samples (coefficients taken as standard normal random variables) of $H_{13}^{(inv)}$ on a block of $l=1,2,3,4,5$ qubits.  The values are ordered with respect to increasing eigenvalue.  Throughout the bulk of the spectrum a value close to the maximum $1-\frac{1}{2^l}$ is seen whereas at the edge of the spectrum a deviation from this is observed.  This is consistent with the area law for low lying eigenstates.

\begin{figure}
  \centering
  \includegraphics{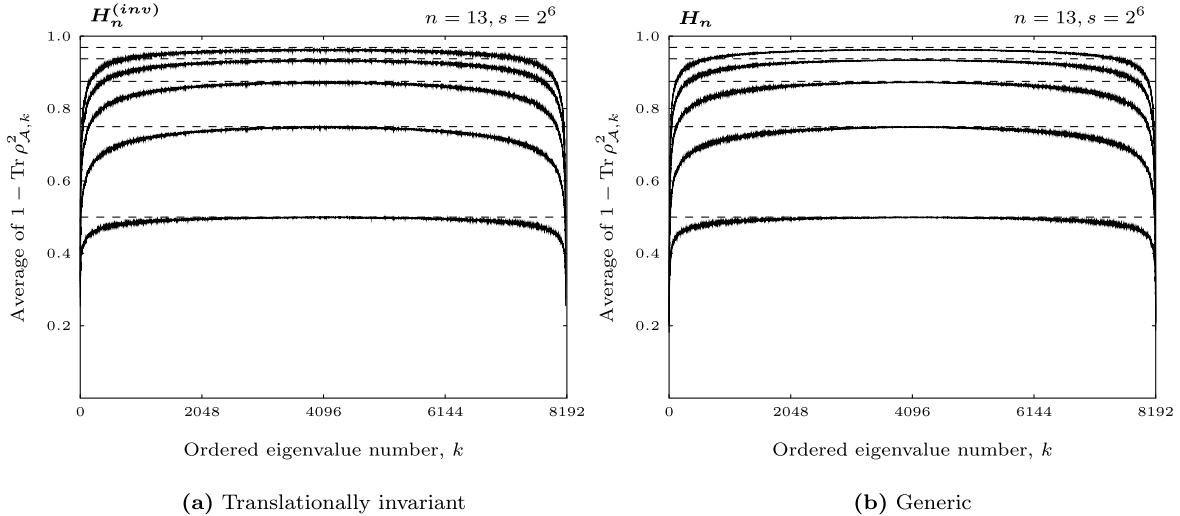}
  \caption{The average value of the linear entropy of the reduced eigenstates over $s=2^6$ random samples (coefficients taken as standard normal random variables) of $H_{13}^{(inv)}$ (Subfigure {\bf (a)}) and $H_{13}$ (Subfigure {\bf (b)}) on a block of $l=1,2,3,4,5$ qubits (solid lines bottom to top).  The dashed lines are at $1-\frac{1}{2^n}$ and give the maximal value of the associated linear entropy values.}
  \label{etang}
\end{figure}

\subsection{Eigenvalue uniqueness}
Non-degeneracy in the eigenvalues of a matrix $H_n^{(inv)}$ is useful for the interpretation of Theorem \ref{th2} with regards to nearest-neighbour Hamiltonians (or more generally in the case of a matrix $G_n^{(inv)}$).  This ensures the uniqueness of the eigenstates of $H_n^{(inv)}$ (up to phase) and therefore in such a case Theorem \ref{th2} refers to the unique eigenstates of a nearest-neighbour spin chain Hamiltonian.  In Section \ref{uniqueness} degeneracy was shown to be the case for generic matrices $H_n^{(inv)}$ for $n$ an odd prime number.  For $4<n<14$ numerical examples of $H_n^{(inv)}$ have been found with a simple spectrum so that the proof of Lemma \ref{exist} can be used to show that this is the case for generic matrices $H_n^{(inv)}$ for these values of $n$.  For larger values of $n$, eigenvalue uniqueness remains open.

The presences of the local terms proportional to $\sigma_j^{(a)}$ plays a crucial part in the simplicity of the spectrum of $H_n^{(inv)}$.  Without them, it was shown in our previous paper \cite{KLW2014_2}, that for odd values of $n$ a Kramers degeneracy exists leading to doubly degenerate eigenvalues throughout the spectrum.

\subsection{Extension of Theorem \ref{th2} to non-translationally invariant matrices}
Figure \ref{etang}b shows the average value of the linear entropy of the reduced eigenstates over $s=2^6$ random samples (coefficients taken as standard normal random variables) of $H_{13}$ on a block of $l=1,2,3,4,5$ qubits.  The values are again ordered with respect to increasing eigenvalue.  It can be shown that the eigenvalues of $H_n$ are also generically non-degenerate (by considering a non-degenerate Hamiltonian of the form $\sum_j\epsilon^j\sigma_j^{(3)}$ for some real $\epsilon$ and applying Lemma \ref{exist}).  A remarkable similarity is seen to the case of $H_n^{(inv)}$ and it is possible that a similar result to that of Theorem \ref{th2} could hold for matrices of the form $H_n$.  The machinery in the proof of Theorem \ref{th2} is not applicable to matrices of this form though, and this question remains open.

\subsection{Tightness of the bound in Theorem \ref{DOS} (Density of states)}
In our related paper \cite{KLW2014_2}, matrices of the form 
\begin{equation}
	H_n^{(JW)}=\frac{1}{\sqrt{\mathcal{C}}}\sum_{j=1}^{n-1}\sum_{a,b=1}^2{\alpha}_{a,b,j}\sigma_{j}^{(a)}\sigma_{  j +1 }^{(b)}+\frac{1}{\sqrt{\mathcal{C}}}\sum_{j=1}^{n}{\alpha}_{3,0,j}\sigma_j^{(3)}
\end{equation}
where the ${\alpha}_{a,b,j}$ are some real constants and $\mathcal{C}$ is the sum of their squares, were diagonalised numerically for values of $n$ up to $32$ using the Jordan-Wigner transform.  Such matrices provide numerical evidence that the bound in Theorem \ref{DOS} is at least asymptotically tight.  In \cite{KLW2014_2} it was seen numerically that the number $\mathcal{N}_n(x)$ of the eigenvalues of a generic instance of $H_n^{(JW)}$, with values less than or equal to $x$, appear to satisfy
\begin{equation}
	\left|\frac{\mathcal{N}_n(x)}{2^n}-\frac{1}{\sqrt{2\pi}}\int_{\infty}^{x}\e^{-\frac{\lambda^2}{2}}\di\lambda\right|\leq \frac{c(x)}{n}
\end{equation}
where $c(x)$ is independent of $n$.  Whether this is the true rate of convergence remains an open question.

\section*{Appendices}
\appendix
\section{Spin chain diagonalisation via the Jordan-Wigner transform}\label{JW}
The Jordan-Wigner transformation \cite{Nielsen2005} gives a standard route to diagonalising the matrix
\begin{equation}
	H_n^{(\epsilon XY+Z)}=\sum_{j=1}^n\left(\epsilon\sigma_{  j  }^{(1)}\sigma_{  j+1  }^{(2)}+\sigma_{  j  }^{(3)}\right)
\end{equation}
We will require this when $n$ takes odd prime values.  The transform is used to define the Fermi creation and annihilation operators
\begin{align}
  a_j=\left(\prod_{1\leq l<j}\sigma_l^{(3)}\right)S_j\qquad\qquad\text{with}\qquad\qquad
  a_j^\dagger=\left(\prod_{1\leq l<j}\sigma_l^{(3)}\right)S_j^\dagger
\end{align}
where
\begin{equation}
  S_j=\frac{\sigma_{  j  }^{(   1   )}+\im\sigma_{  j  }^{(   2   )}}{2}\qquad\qquad\text{with}\qquad\qquad
  S_j^\dagger=\frac{\sigma_{  j  }^{(   1   )}-\im\sigma_{  j  }^{(   2   )}}{2}
\end{equation}
so that for $j=1,\dots,n-1$, as seen in \cite{Nielsen2005},
\begin{align}
	\sigma_{  j  }^{(1)}\sigma_{  j+1  }^{(2)}&=\im\left(a_j-a_j^\dagger\right)\left(a_{j+1}-a_{j+1}^\dagger\right)\nonumber\\
	\sigma_{  n  }^{(1)}\sigma_{  1  }^{(2)}&=-\im\eta\Big(a_n-a_n^\dagger\Big)\left(a_1-a_1^\dagger\right)\nonumber\\
	\sigma_{  j  }^{(3)}&=a_ja_j^\dagger-a_j^\dagger a_j
\end{align}
The canonical commutation relations for Fermi operators, $a_ja_k=-a_ka_j$ and $a_ja_k^\dagger=-a_k^\dagger a_j+I\delta_{jk}$, can also be verified to hold.  This allows $H_n^{(\epsilon XY+Z)}$ to be rewritten as
\begin{equation}
	H_n^{(\epsilon XY+Z)}=\im\epsilon\sum_{j=1}^{n-1}\left(a_j-a_j^\dagger\right)\left(a_{j+1}-a_{j+1}^\dagger\right)-\im\epsilon\eta\Big(a_n-a_n^\dagger\Big)\left(a_1-a_1^\dagger\right)+\sum_{j=1}^n\left(a_ja_j^\dagger-a_j^\dagger a_j\right)
\end{equation}

It will now be shown that $H_n^{(\epsilon XY+Z)}$ is block diagonal in the product basis formed from the $n$-fold tensor product of eigenstates of $\sigma^{(3)}=|0\rangle\langle0|-|1\rangle\langle1|$, denoted $|0\rangle$ and $|1\rangle$.  Let this basis be denoted $|\boldsymbol{x}\rangle_z=|x_1\rangle\otimes\dots\otimes|x_n\rangle$ for some multi-index $\boldsymbol{x}=(x_1,\dots,x_n)\in\{0,1\}^n$.  The Hamiltonian $H_n^{(\epsilon XY+Z)}$ commutes with the operator
\begin{equation}
	\eta=\prod_{j=1}^{n}\sigma_{  j  }^{(3)}
\end{equation}
which has two eigenvalues $\pm1$, this is most easily seen when $H_n^{(\epsilon XY+Z)}$ is expressed in the Pauli basis. The operator $\eta$ is diagonal in the basis $|\boldsymbol{x}\rangle_z$ so that $H_n^{(\epsilon XY+Z)}$ must be block diagonal in this basis with two blocks corresponding to the two eigenvalues of $\eta$.  The diagonalisation will now be performed on each block, labelled by $\eta=\pm1$, separately.

\subsubsection*{The $\eta=-1$ block}
For states in the $\eta=-1$ subspace of the Hilbert space, $H_n^{(\epsilon XY+Z)}$ acts as
\begin{equation}
	H_n^{-}=\im\epsilon\sum_{j=1}^{n}\left(a_j-a_j^\dagger\right)\left(a_{j+1}-a_{j+1}^\dagger\right)+\sum_{j=1}^n\left(a_ja_j^\dagger-a_j^\dagger a_j\right)
\end{equation}
with the periodic boundary conditions $a_{j+n}=a_j$ imposed.  Again, $H_n^-$ is seen to commute with the operator $\eta$.  Using the canonical commutation relations for Fermi operators, this quadratic form in the Fermi operators $\{a_j,a_j^\dagger\}$ can be expressed as
\begin{equation}
	H_n^{-}=
	\begin{pmatrix}
		\boldsymbol{a}^\dagger&
		\boldsymbol{\tilde{a}}^\dagger
	\end{pmatrix}
	\begin{pmatrix}
		A-I&-A\\
		-A&A+I
	\end{pmatrix}
	\begin{pmatrix}
		\boldsymbol{a}\\
		\boldsymbol{\tilde{a}}
	\end{pmatrix}
\end{equation}
where $\boldsymbol{a}$ is the column vector with entries $a_1,\dots,a_n$, $\boldsymbol{\tilde{a}}$ is the column vector with entries $a_1^\dagger,\dots,a_n^\dagger$, $\boldsymbol{a}^\dagger$ is the row vector with entries $a_1^\dagger,\dots,a_n^\dagger$, $\boldsymbol{\tilde{a}}^\dagger$ is the row vector with entries $a_1,\dots,a_n$ and
\begin{equation}
	A=\frac{\im\epsilon}{2}
	\begin{pmatrix}
		0&-1&0&\cdots&0&1\\
	  1&0&-1&\ddots&&0\\
	  0&1&\ddots&\ddots&&\vdots\\
	  \vdots&\ddots&\ddots&&&0\\
	  0&&&&0&-1\\
	  -1&0&\cdots&0&1&0
	\end{pmatrix}
\end{equation}

A further set of Fermi operators $\{b_j, b_j^\dagger\}$ can be defined by the unitary transform
\begin{equation}
	\begin{pmatrix}
		\boldsymbol{a}\\
		\boldsymbol{\tilde{a}}
	\end{pmatrix}
	=
	\begin{pmatrix}
		U&0\\
		0&\overline{U}
	\end{pmatrix}
	\begin{pmatrix}
		\boldsymbol{b}\\
		\boldsymbol{\tilde{b}}
	\end{pmatrix}
\end{equation}
for the unitary matrix $U_{jk}=\frac{1}{\sqrt{n}}\omega_{j}^{k}$ where $\omega_j=\e^{\frac{2\pi\im j}{n}}$ (the operator $\{b_j, b_j^\dagger\}$ are guaranteed to be Fermi operators \cite{Nielsen2005}). In fact, $U$ represents a discrete periodic Fourier transform and $A$ is proportional to a circulant matrix.  By \cite[p.388]{circulant}, $U$ diagonalises $A$ so that $U^\dagger AU=D$ where $D_{jk}=\delta_{jk}\epsilon\sin\left(\frac{2\pi j}{n}\right)$.  Furthermore, $\overline{U}_{jk}=\frac{1}{\sqrt{n}}\omega_{j}^{-k}$ so that if $P_{jk}=\delta_{j\kappa(k)}$ with $\kappa(k)=n-k$ for $k=1,\dots,n-1$ and  $\kappa(n)=n$ then
\begin{equation}
  U^\dagger A\overline{U}=U^\dagger AUP=DP
\end{equation}
The Hamiltonian $H^-$ is therefore transformed to
\begin{equation}
	\begin{pmatrix}
		\boldsymbol{b}^\dagger&
		\boldsymbol{\tilde{b}}^\dagger
	\end{pmatrix}
	\begin{pmatrix}
		U^\dagger AU-I&-U^\dagger A\overline{U} \\
		\overline{U^\dagger A\overline{U}}&-\overline{U^\dagger AU}+I
	\end{pmatrix}
	\begin{pmatrix}
		\boldsymbol{b}\\
		\boldsymbol{\tilde{b}}
	\end{pmatrix}
	=
	\begin{pmatrix}
		\boldsymbol{b}^\dagger&
		\boldsymbol{\tilde{b}}^\dagger
	\end{pmatrix}
	\begin{pmatrix}
		D-I&-DP \\
		DP&-D+I
	\end{pmatrix}
	\begin{pmatrix}
		\boldsymbol{b}\\
		\boldsymbol{\tilde{b}}
	\end{pmatrix}
\end{equation}

Due to the form of $D$ and $DP$, $H_n^{-}$ can now be rewritten as
\begin{equation}
	\sum_{j=1}^n
	\begin{pmatrix}
		b_j^\dagger&
		b_{\kappa(j)}
	\end{pmatrix}
	\begin{pmatrix}
		\epsilon\mu_j-1&-\epsilon\mu_j\\
		\epsilon\mu_{\kappa(j)}&-\epsilon\mu_{\kappa(j)}+1
	\end{pmatrix}
	\begin{pmatrix}
		b_j\\
		b_{\kappa(j)}^\dagger
	\end{pmatrix}
\end{equation}
with $\mu_j=\sin\left(\frac{2\pi j}{n}\right)$.  Each of the $2\times 2$ matrices in this expression may be diagonalised by a unitary matrix, therefore implying a further unitary transformation from the $\{b_j,b_j^\dagger\}$ to new Fermi operators $\{c_j,c_j^\dagger\}$ (again a unitary matrix may be chosen such that the operators $\{c_j,c_j^\dagger\}$ are again Fermi operators \cite{Nielsen2005}).  As $\mu_j=-\mu_{\kappa(j)}$, the matrix
\begin{equation}
	\begin{pmatrix}
		\epsilon\mu_j-1&-\epsilon\mu_j\\
		\epsilon\mu_{\kappa(j)}&-\epsilon\mu_{\kappa(j)}+1
	\end{pmatrix}=
	\begin{pmatrix}
		\epsilon\mu_j-1&-\epsilon\mu_j\\
		-\epsilon\mu_{j}&\epsilon\mu_{j}+1
	\end{pmatrix}
\end{equation}
has the eigenvalues $\chi_j^{\pm}=\epsilon\mu_j\pm\sqrt{\epsilon^2\mu_j^2+1}$ so that
\begin{equation}
	H_n^{-}
	=\sum_{j=1}^n\left(\chi_j^{-}c_j^\dagger c_j+\chi_j^{+}c_{\kappa(j)}c_{\kappa(j)}^\dagger\right)
	=\sum_{j=1}^n\left(\chi_j^{-}c_j^\dagger c_j+\chi_{\kappa(j)}^{+}c_jc_j^\dagger\right)
\end{equation}
Then the Hilbert space admits the orthonormal Fermi basis
\begin{equation}
	|\boldsymbol{x}\rangle_c
	=\Big(c_1^\dagger\Big)^{x_1}\dots\Big(c_n^\dagger\Big)^{x_n}|\boldsymbol{0}\rangle_c
\end{equation}
for the multi-index $\boldsymbol{x}=(x_1,\dots,x_n)\in\{0,1\}^n$, where $|\boldsymbol{0}\rangle_c$ is a normalised state such that $c_j|\boldsymbol{0}\rangle_c=0$ for all $j$.  These states are exactly the eigenstates of $H_n^{-}$ and the corresponding eigenvalues are read off to be
\begin{equation}
  \lambda_{\boldsymbol{x}}=\sum_{j=1}^n\left(\left(\epsilon\mu_j-\sqrt{\epsilon^2\mu_j^2+1}\right)x_j+\left(\epsilon\mu_{\kappa(j)}+\sqrt{\epsilon^2\mu_{\kappa(j)}^2+1}\right)(1-x_j)\right)
\end{equation}
Substituting $\mu_j=-\mu_{\kappa(j)}$ simplifies this expression to
\begin{equation}
	\lambda_{\boldsymbol{x}}=\sum_{j=1}^n(2x_j-1)\left(\epsilon\mu_j-\sqrt{\epsilon^2\mu_j^2+1}\right)
\end{equation}

It has already been seen in the proof of Lemma \ref{dist} that the values $\lambda_{\boldsymbol{x}}$ are distinct for most values of $\epsilon>0$ in the case that $n$ is an odd prime.  So for such values, as $H_n^{-}$ commutes with $\eta$, the eigenstates $|\boldsymbol{x}\rangle_c$ of $H_n^{-}$ are also eigenstates of $\eta$ with eigenvalues $\pm1$.  Only the eigenvalues $\lambda_{\boldsymbol{x}}$ corresponding to eigenstates for which $\eta|\boldsymbol{x}\rangle_c=-|\boldsymbol{x}\rangle_c$ are also eigenvalues of $H_n^{(\epsilon XY+Z)}$, the others have no relevance to the spectrum of $H_n^{(\epsilon XY+Z)}$ as concern lies with the $\eta=-1$ eigenspace at this point.

By expressing $c_j^\dagger$ in terms of the the Pauli operators (inverting the previous transforms they are a linear combination of the Fermi operators $\{a_j,a_j^\dagger\}$), it is seen that $\eta$ and $c_j^\dagger$ anti-commute for each $j$ so that $\eta|\boldsymbol{x}\rangle_c=(-1)^r\big(c_1^\dagger\big)^{x_1}\dots\big(c_n^\dagger\big)^{x_n}\eta|\boldsymbol{0}\rangle_c$ where $r=\sum_jx_j$.  Now as $\eta|\boldsymbol{0}\rangle_c=\pm|\boldsymbol{0}\rangle_c$, either all the states $|\boldsymbol{x}\rangle_c$ for which $r$ is even (if $\eta|\boldsymbol{0}\rangle_c=-|\boldsymbol{0}\rangle_c$) or odd (if $\eta|\boldsymbol{0}\rangle_c=+|\boldsymbol{0}\rangle_c$) are eigenstates of $H_n^{(\epsilon XY+Z)}$.

\subsubsection*{The $\eta=1$ block}
The spectrum within the the $\eta=1$ subspace for odd prime values of $n$ can be deduced by symmetry.  The operator $\nu=\sum_{j=1}^n\sigma_{  j  }^{(1)}$ anti-commutes with $H_n^{(\epsilon XY+Z)}$ as $\sigma^{(1)}$ anti-commutes with both $\sigma^{(2)}$ and $\sigma^{(3)}$.  Any eigenstate $|\psi\rangle$ of $H_n^{(\epsilon XY+Z)}$ in the $\eta=-1$ subspace with eigenvalue $\lambda$ must then satisfy
\begin{equation}
	H_n^{(\epsilon XY+Z)}\left(\nu|\psi\rangle\right)=-\nu H_n^{(\epsilon XY+Z)}|\psi\rangle=-\lambda\left(\nu|\psi\rangle\right)
\end{equation}
The eigenstate $\nu|\psi\rangle$ of $H_n^{(\epsilon XY+Z)}$ must be in the $\eta=1$ subspace as $n$ is odd and therefore $\eta\nu=-\nu\eta$ so that $\eta\nu|\psi\rangle=-\nu\eta|\psi\rangle=\nu|\psi\rangle$.  This implies that the spectrum within the $\eta=1$ subspace is the negative of that within the $\eta=-1$ subspace.  This is precisely given by the values
\begin{equation}
	\lambda_{\boldsymbol{x}}=\sum_{j=1}^n(2x_j-1)\left(\epsilon\mu_j-\sqrt{\epsilon^2\mu_j^2+1}\right)
\end{equation}
indexed by the $\boldsymbol{x}$ for which $r=\sum_jx_j$ has the opposite parity as taken for the $\eta=-1$ block.

The entire spectrum for odd prime values of $n$ and most $\epsilon>0$ (and therefore all $\epsilon$ by the analyticity of the eigenvalues) is therefore given by
\begin{equation}
	\lambda_{\boldsymbol{x}}=\sum_{j=1}^n(2x_j-1)\left(\epsilon\mu_j-\sqrt{\epsilon^2\mu_j^2+1}\right)
\end{equation}

\footnotesize
\bibliographystyle{cmp}
\bibliography{SpinChains}

\end{document}